\documentclass{article}
\usepackage{graphicx} % Required for inserting images
\usepackage[utf8]{inputenc}
\usepackage{enumerate}
\usepackage[inline]{enumitem}

\usepackage{fullpage}
\usepackage{amsmath,amssymb,amsthm,hyperref,verbatim}
\usepackage{graphicx}
\usepackage{xcolor}
\usepackage[capitalize]{cleveref}
\usepackage{authblk}
\usepackage{delimset}
\newtheorem{theorem}{Theorem}
\newtheorem{claim}[theorem]{Claim}

\newtheorem{lemma}[theorem]{Lemma}

\newtheorem{corollary}[theorem]{Corollary}
\newtheorem{definition}[theorem]{Definition}

\usepackage{algorithm}
\usepackage{algpseudocode}

\newcommand{\eps}{\epsilon}

\newcommand{\init}{\mathsf{init}}
\newcommand{\fail}{\mathsf{fail}}
\DeclareMathOperator{\op}{\mathsf{op}}
\DeclareMathOperator{\ins}{\mathsf{insert}}
\DeclareMathOperator{\del}{\mathsf{delete}}
\DeclareMathOperator{\query}{\mathsf{query}}
\DeclareMathOperator{\fspace}{\mathsf{space}}
\DeclareMathOperator{\FP}{\mathsf{FP}}
\DeclareMathOperator{\FN}{\mathsf{FN}}
\DeclareMathOperator{\E}{\mathbb{E}}
\newcommand{\pfail}{p_{\mathsf{fail}}}
\newcommand{\pfp}{\eps^{+}}
\newcommand{\pfn}{\eps^{-}}
\newcommand{\Fs}{F_{\mathsf{static}}}
\newcommand{\Fd}{F_{\mathsf{dynamic}}}

\providecommand{\keywords}[1]
{
  \small	
  \textbf{\textit{Keywords---}} #1
}
\title{A Space Lower Bound for Approximate Membership with
Duplicate Insertions or Deletions of Nonelements}

\author[1]{Aryan Agarwala}
\author[2]{Guy Even\thanks{guy@eng.tau.ac.il}}
\affil[1]{Max Planck Institute for Informatics, Saarland Informatics Campus}
\affil[2]{School of Electrical Engineering, Tel Aviv University. Work done while visiting the Max Planck Institute for Informatics, Saarland Informatics Campus}
\date{\today}

\begin{document}
\maketitle
\begin{abstract}
Designs of data structures for approximate membership queries with false-positive errors that support both insertions and deletions stipulate the following two conditions: 
\begin{enumerate*}[label=(\arabic*)]
\item Duplicate insertions are prohibited, i.e., it is prohibited to insert an element $x$ if $x$ is currently a member of the dataset. 
\item Deletions of nonelements are prohibited, i.e., it is prohibited to delete $x$ if $x$ is not currently a member of the dataset.
\end{enumerate*}
Under these conditions, the space required for the approximate representation of a datasets of cardinality $n$ with a false-positive probability of $\pfp$ is at most $(1+o(1))n\cdot\log_2 (1/\pfp) + O(n)$ bits [Bender et al., 2018; Bercea and Even, 2019]. 

We prove that if these conditions are lifted, then the space required for the approximate representation of datasets of cardinality $n$ from a universe of cardinality $u$ is at least $\frac 12 \cdot (1-\pfp -\frac 1n)\cdot \log \binom{u}{n} -O(n)$ bits.
\end{abstract}
\section{Introduction}
Dynamic data structures for approximate membership store an approximate representation of a dataset subject to insert and delete operations. The output of such data structures for membership queries may be wrong, specifically false-positive errors are allowed.

We refer to such data structures in short as \emph{dynamic filters}. \footnote{Static filters store an approximate representation of a fixed dataset. Incremental filters support insertions but not deletions. We refer to static/incremental/dynamic filters simply as filters.}
Filters are randomized data structures, and the probability of a false-positive error is bounded by a (tunable) parameter denoted by $\pfp$. The main advantage of the approximate representation used by a filter over exact error-free representation is the reduction in space (i.e., number of bits needed to store the representation). 
Indeed, consider a universe of cardinality $u$ and datasets of cardinality $n$. 
Because there are $\binom{u}{n}$ different subsets of $u$ cardinality $n$, it follows that exact representation of such datasets requires at least $\log \binom{u}{n}\geq n\log (u/n)$ bits.
\footnote{All logarithms in this paper are base $2$.}
On the other hand, dynamic filters with a false-positive error probability of at most $\pfp$ require only $(1+o(1))n\log(1/\pfp)+O(n)$ bits~\cite{bender2018bloom,bercea2019fully,Bercea2024OneMemeoryAccess}. 
If $\log(1/\pfp)$ is much smaller than $\log(u/n)$, then the space of a filter is smaller than the space required for exact representation. Such a reduction in space often enables the storage of the approximate representation in fast yet small memory~\cite{Bloom1970filter} which is advantageous in many applications (see~\cite{broder2004network} for a survey).

Dynamic filter designs, such as Cuckoo filters~\cite{fan2014cuckoo,eppstein2016cuckoo}, 
Quotient filters~\cite{pagh2005optimal,bender2011thrash}, Broom filters~\cite{bender2018bloom}, and~\cite{bercea2019fully} are based on the framework suggested by Carter et al.~\cite{carter1978approximate}.
Carter et al. presented a general framework for incremental filters that support insertions and approximate membership queries. To obtain a false-positive error probability of at most $\pfp$, the framework hashes the $n$ keys to strings of $\log (n/\pfp)$ bits. The images of the keys are called \emph{fingerprints} and are stored in a dictionary (i.e., hashtable) that supports exact representation of datasets~\cite{arbitman2010backyard}. That is, an $\ins(x)$ operation to the filter translates to an $\ins(\FP(x))$ to the dictionary, where $\FP(x)$ is the fingerprint of $x$.  A query for $x$ simply searches for $\FP(x)$ in the dictionary. To support deletions in this framework, the dictionary needs to store multisets of fingerprints~\cite{pagh2005optimal} (i.e., the multiplicity of a fingerprint equals the number of elements in the dataset that map to that fingerprint). When a $\del(x)$ operation is issued to the filter, the multiplicity of $\FP(x)$ in the dictionary is decremented (e.g., deletion of a copy of $\FP(x)$).

\medskip\noindent
For a dataset $S$, we refer to $x\in S$ as an \emph{element} and to $x\not\in S$ as a \emph{nonelement}. Previous designs of dynamic filters~\cite{fan2014cuckoo,pagh2005optimal,bender2011thrash,bender2018bloom,bercea2019fully,even2024micro} stipulate the following conditions:
\begin{enumerate}
    \item \textbf{Duplicate insertions are prohibited}. One may not perform an $\ins(x)$ operation if $x$ is an element (i.e., already an element in the dataset). 
    \item \textbf{Deletions of nonelements are prohibited.} One may not  perform a $\del(x)$ operation if $x$ is a nonelement (i.e. not an element in the dataset).
\end{enumerate}
These two conditions arise from the limitations of the framework for designing a dynamic filter by storing a multiset of fingerprints.  The reason for prohibiting duplicate insertions is the inability of the framework to distinguish between a duplicate insertion of an element $x$ and an insertion of a nonelement $x$ whose fingerprint collides with a fingerprint of an element $y$. Observe that these two events should have a different influence on the approximate representation of the dataset: 
\begin{enumerate*}[label=(\arabic*)]
    \item An $\ins(x)$ operation, where $x$ is an nonelement, $y$ is an element, and $\FP(y)=\FP(x)$ should increment the multiplicity of $\FP(x)$. Otherwise, a subsequent $\del(y)$ operation (that decrements the multiplicity of $\FP(y)$) would zero the multiplicity of $\FP(y)$ and thus introduce a false-negative error for $\query(x)$. 
    \item A duplicate insertion $\ins(x)$ does not modify the dataset and should be ignored. In other words, if the multiplicity of $\FP(x)$ is incremented, then a $\del(x)$ followed by $\query(x)$ would definitely return a false-positive error.
\end{enumerate*}

The reason for prohibiting a deletion of a nonelement $x$ is that there might exist an element $y$ such that $\FP(x)=\FP(y)$. Because the framework cannot distinguish between a deletion of $x$ and a deletion of $y$, a $\del(x)$ operation would delete $\FP(y)$, leading to a false-negative error for $\query(y)$.

In this paper we address the following question: \textbf{Is it possible to design a small space dynamic filter that permits duplicate-insertions or deletions-of-nonelements?}
We answer this question negatively (see \Cref{thm:dupilicate ins deletion of nonelement} for a precise formulation), and prove that if a dynamic filter permits duplicate-insertions or deletions-of-nonelements, then the space required to store the approximate representation of the dataset is $\Omega\brk*{\binom{u}{n}}$ bits (if $u\geq 2n$).
This means that the space of such a filter is at least a constant fraction of the space required for exact representation (i.e., the fraction is at least $(1-\pfp-o(1))/2$ if deletions-of-nonelements are allowed). The lower bound does not make assumptions about the running time of the filter or the number of memory accesses. Moreover, the lower bound does not count the random bits of the filter or temporary memory used for processing  operations.
The lower bound holds with respect to an oblivious adversary.

\subsection{Related Work}

\subsubsection{Upper Bounds}
The first filter that supports insertions and approximate membership queries (but no deletions) was presented by Bloom~\cite{Bloom1970filter}. Carter et al~\cite{carter1978approximate} presented a general framework for filters that support insertions and approximate membership queries. The framework suggests to hash the keys to strings of length $\log (n/\pfp)$ using a family of universal hash functions. The images of the keys are called \emph{fingerprints} and are stored in a dictionary (i.e. hashtable) that supports exact representation of datasets. A query for $x$ simply searches for its fingerprint in the dictionary. Arbitman et al.~\cite{arbitman2010backyard} presented a succinct constant-time dictionary. Storing fingerprints in  such a dictionary yields a constant-time filter that supports insertions and approximate membership queries with space $(1+o(1))n\log (1/\pfp) + O(n)$ bits.
Collisions of fingerprints pose a challenge in the design of dynamic filters that support deletions in addition to insertions and approximate membership queries. Dynamic filters with space $(1+o(1))n\log (1/\pfp) +O(n)$ bits are described in~\cite{bender2018bloom,bercea2019fully,Bercea2024OneMemeoryAccess}. 
Fan et al.~\cite{fan2014cuckoo} presented the Cuckoo Filter that stores fingerprints in a Cuckoo hashtable. All the designs for dynamic filters prohibit duplicate-insertions and deletions-of-nonelements.

\subsubsection{Lower Bounds}
The following lower bounds related to filters have been proven to date.
\begin{enumerate}
    \item A static filter for datasets of cardinality $n$ with false-positive error probability $\pfp$ requires space that is at least $n\log (1/\pfp)-O(1)$ bits (if $u>n^2/\pfp$). ~\cite{carter1978approximate,broder2004network,dietzfelbinger2008succinct}.
\item An extendable dynamic filter whose space is a function of $n=|S|$, the current cardinality of the dataset, requires space that is at least $(1-o(1))\cdot n\cdot \brk*{\log \frac{1}{\pfp}+ (1-O(\pfp))\log\log n}$ bits~\cite{PaghSegevWieder2013size}.
\item For every constant $\pfp>0$, an incremental filter requires space that is at least 
$C(\pfp)\cdot n\log (1/\pfp)$ bits, where $C(\pfp)>1$~\cite{LovettPorat2013lowerbound}

\item A dynamic filter that is resilient to an adaptive adversary that tries to repeat a false-positive error requires space that is $\Omega(\min\{n\log\log u, n \log n\})$ bits~\cite{bender2018bloom}.

\item A dynamic filter requires space that is $n\log(1/\pfp) + \Omega(n)$ (even if $\pfp=o(1)$)~\cite{kuszmaul2024space}.
\end{enumerate}

\subsection{Outline of Proof}
In this section, we outline the proof of the lower bound for the space of  dynamic filters that permit duplicate-insertions or deletions-of-nonelements (\Cref{thm:dupilicate ins deletion of nonelement}) in simplified terms; full details and a formal presentation appear in the paper.

The starting point of the proof is a lower bound on a static filter with false-negative errors but without false-positive errors. Carter et al.~\cite{carter1978approximate} write:\footnote{The nomenclature in Carter et. al refers to false-negative errors as ``false alarms'' and to a dataset as a ``vocabulary''.} \emph{``It is easy to see that little memory can be saved by permitting false alarms; a membership tester with a small false alarm probability is in fact a checker for a slightly smaller vocabulary.''} We failed to find a proof of this statement in the literature, so we formalize this statement and prove  in \Cref{thm:static filter} and \Cref{coro:pfp=0} that the space of a static filter with false-negative errors (but without false-positive errors) is $\Omega(\log \binom{u}{n})$.

In \Cref{def:witness-based}, we define a \emph{witness-based filter} $F^*$. Given an approximate representation $M$, the witness-based filter $F^*$ answers ``yes'' to $\query(x)$ if only if there exists a dataset $S$ such that $x\in S$ and $M$ is an approximate representation of $S$. 
One can transform every dynamic filter $F$ to a witness-based dynamic filter $F^*$ by scanning all the datasets while preserving the parameters of $F$ (i.e., space as well as the bound on false-positive error probability). Because our lower bound ignores running time and temporary memory, we may focus, without loss of generality, on witness-based dynamic filters.

\Cref{lemma:false positives} proves that the false-positive errors in a witness-based dynamic filter that permits deletions-of-nonelements are ``sticky''. That is, if $S$ is a dataset, $x\not \in S$, and $\query(x)$ returns ``yes'' after the elements of $S$ are inserted, then $\query(x)$ still returns ``yes'' after the elements of $S$ are deleted. 

In \Cref{def:reduction}, we construct a static filter $\sigma(F)$ with false-negative errors (but no false-positive errors) by using a witness-based dynamic filter $F$ that permits deletions-of-nonelements as a black box. 
Let $S$ denote the dataset that the static filter $\sigma(F)$ should approximately represent. The reduction uses two approximate dataset representations by  $F$: 
\begin{enumerate*}[label=(\arabic*)]
\item The approximate representation $M(+S)$ of $S$ by $F$.
\item The approximate representation $M(+S,-S)$ by $F$ obtained after deleting the elements of $S$ from $M(+S)$.
\end{enumerate*}
The static filter $\sigma(F)$ processes a query for $x$ by applying two queries to $F$: one with the representation $M(+S)$ that returns $b_1$, and the second with representation $M(+S,-S)$ that returns $b_2$. The answer is $b\triangleq(b_1 \text{ and } \mathsf{not}(b_2))$. In \Cref{lemma:reduction}, a simple case analysis proves that $\sigma(F)$ lacks false-positive errors and the false-negative probability is bounded by the false-positive probability of $F$.
The proof of the lower bound on the space of $F$ (i.e., \Cref{thm:dupilicate ins deletion of nonelement}) is completed by applying the lower bound on the space of a static filter with false-negative errors to $\sigma(F)$ (i.e., \Cref{coro:pfp=0}).

We present simple reductions between duplicate-insertions and deletions-of-non-elements (in both directions) in \Cref{claim:reductions}. Hence, it suffices to consider only deletions-of-non-elements. 
%%%%%%%%%%%%%%%%%%%%
\section{Preliminaries}
\subsection{Notation}
For a nonnegative integer $k$,  let $[k]$ denote the set $\{0,1,\ldots, k-1\}$. We refer to the set from which elements are as the \emph{universe}. Without loss of generality, we assume that the universe is $[u]$, where $u\in\mathbb{N}$ is the cardinality of the universe. 
A \emph{dataset} is a subset of $[u]$.
Let $n$ denote an upper bound on the cardinality of datasets. 
All logarithms in this paper are base two.

\subsection{Randomized Algorithms}
The data-structures in this paper are randomized. 
For simplicity, we view randomization as drawing a  string $r$ of random bits when the data-structure in initialized, where the length of $r$ is bounded by a function of the size of the universe and the size of the dataset. 
One can extend the definition of randomization and the proof of to an unbounded number of bit flips (e.g., new random bits per query or Las-Vegas algorithms that repeat until they succeed).

\subsection{Oblivious Adversary}
Throughout this paper, we assume that the distribution from which the input to the filter is drawn (i.e., choice of dataset and sequence of operations) is independent of the random string $r$ chosen by the filter. 
This independence is referred to as an \emph{oblivious adversary}. 
For example, an adversary is oblivious if the sequence of operations is fully determined before the random string $r$ is chosen by the filter. We emphasize that the lower bound holds even with respect to such a restricted adversary.

\subsection{N-Static Filters: Static Filters with False-Negative Errors}
In this section, we define static filters with false-negative errors (but without false-positive errors) called N-static filters. We point out that the sole purpose of N-static filters is for the proof of the lower bound. 

We begin with a syntactic definition (i.e., API) of static filter in \Cref{def:N-static syntax}.
We then define the semantics (i.e., functionality) of N-static filters in \Cref{def:N-static}.

\begin{definition}[Syntax of Static Filter]\label{def:N-static syntax}
A static filter is a pair of randomized algorithms $\{A_{\init},A_{\query}\}$ with shared randomness. Let $r$ denote the shared string of random bits.
The input to a static filter consists of two parts:
\begin{enumerate*}[label=(\arabic*)]
    \item a dataset $S\subseteq [u]$ , and 
    \item a sequence of queries $\{\query(x_t)\}_{t}$, where $x_t\in [u]$ for every $t$.  
\end{enumerate*}

\medskip\noindent
Processing of operations. 
\begin{enumerate}
    \item Initialization.
The filter is initialized by drawing a random string $r$ and invoking $A_{\init}(r,S)$ that returns the \emph{initial state} $M(r,S)$.  The state $M(r,S)$ is a binary string that serves as an approximate representation of the dataset $S$. 
\item 
A $\query(x_t)$ is processed by invoking $A_{\query}(r,M(r,S),x)$ that returns a bit $b_t$.\footnote{Note that the processing of a query does not modify the state of the filter. Such a filter is referred to in~\cite{Naor-Yogev2019adversarial} as a filter that has a \emph{steady representation.} }
\end{enumerate}
\end{definition}
 
The following definition defines the family of N-static filters $\Fs(u,n,\pfail,\pfn)$ that have two types of errors: initialization errors and false-negative errors. Initialization failure captures failure in some Monte-Carlo constructions for static filters~\cite{dietzfelbinger2008succinct,porat2009optimal,dillinger2021ribbon}. 
The probability space is over the random choice of the string $r$.

\begin{definition}[Semantics of N-Static Filter]\label{def:N-static}
A static filter $F$ is an N-static filter in $\Fs(u,n,\pfail,\pfn)$ 
if for every dataset $S\subseteq [u]$ of cardinality at most $n$ and every sequence $\{\query(x_t)\}_{t}$ of queries, the following properties hold:
   \begin{enumerate}
       \item Initialization Success. $\Pr_r[M(r,S)=\fail] \leq \pfail$. 
       \item Soundness (bounded false-negative errors). If $x_t\in S_t$, then $\Pr_r [b_t=0\mid M(r,S)\neq\fail]\leq \pfn$.
       \item Completeness (no false-positive errors). If $x_t\not\in S_t$ or $M(r,S)=\fail$, 
       then $b_t=0$.
       \end{enumerate}
\end{definition}
To avoid false-positive errors, an N-static filter responds to every query with ``no'' (i.e., $b_t=0$) if the state is $\fail$.
\begin{definition}[Space of Static Filter]
The space used by a static filter $F$ is the maximum length (in bits) of the initial state $M(r,S)$ over all subsets $S\subseteq [u]$ of cardinality at most $n$ and all strings $r$. Formally,
\begin{align*}
     \fspace(F)&\triangleq\max_{S\subseteq [n]: |S|=n}\max_r |M(r,S)|\;.
\end{align*}\end{definition}
We emphasize that the string $r$ of random bits is not part of the state and is not counted in the space. Indeed, $r$ is input to both $A_{\init}$ and $A_{\query}$  and serves as shared randomness (i.e., public coins). 
Moreover, the space does not include internal memory used by the filter's algorithms (i.e.,
temporary memory that is used for the processing of an operation).
\footnote{Excluding internal memory and the random string from the space does not weaken the lower bound.}

\subsection{Dynamic Filters}
\begin{definition}[Sequences of Operations and Datasets]
A sequence $\{\op_t\}_{t=0}^T$  of initialize/insert/delete/query operations for which $\op_0$ is an initialization operation inductively defines the sequence of sets $\{S_t\}_{t=0}^T$   as follows:
    \begin{enumerate}
            \item If $\op_t=\init$, then $S_t\gets \emptyset$.
        \item If $\op_t=\ins(x_t)$, then $S_t\gets S_{t-1}\cup \{x_t\}$.
            \item If $\op_t=\del(x_t)$, then $S_t\gets S_{t-1}\setminus\{x_t\}$.
            \item If $\op_t=\query(x_t)$, then $S_t\gets S_{t-1}$.   
    \end{enumerate}
\end{definition}
\begin{definition}[Duplicate-Insertions and Deletion-of-Nonelements]
    We refer to $\op_t$ as a \emph{duplicate-insertion} if $x_t\in S_{t-1}$ and $\op_t=\ins(x_t)$.
    We refer to $\op_t$ as a \emph{deletion-of-a-nonelement} if $x_t\not\in S_{t-1}$ and $\op_t=\del(x_t)$.  
\end{definition}
We emphasize that if $\op_t$ is a duplicate-insertion or a deletion-of-a-nonelement, then $S_t=S_{t-1}$.
The following definition of $(u,n)$-sequences permits duplicate-insertions or deletions-of-nonelements. 
\begin{definition}[$(u,n)$-Sequences]
 A sequence $\{\op_t\}_{t=0}^T$  of insert/delete/initialize/query operations is a $(u,n)$-sequence if 
 \begin{enumerate}
     \item $\op_0=\init$, 
     \item the argument of every insert/delete/query operation is in $[u]$ and
     \item $\max_t |S_t|\leq n$. 
 \end{enumerate}
\end{definition}
%%%
The following definition defines a dynamic filter that supports initialization, insertions, deletions, and (approximate) membership queries.
\begin{definition}[Syntax of a Dynamic Filter]
A dynamic filter is a $4$-tuple of randomized algorithms $\{A_{f}\}$, where $f\in\{\init,\ins,\del,\query\}$. 
Let $r$ denote the shared string of random bits used by the algorithms. The input consists of a $(u,n)$-sequence of operations $\{\op_t\}_{t=0}^T$.  Let $M_t$ be a binary string that denotes the \emph{state at time $t$}.

\medskip\noindent
Processing of operations. 
    \begin{enumerate}
        \item An $\op_t=\init$ is processed by drawing the string $r$ of random bits and invoking $M_t\gets A_{\init}(r)$
        \item 
        For $f\in\{\ins,\del\}$,
        an $\op_t=f(x_t)$ is processed by invoking $M_t\gets A_f(r,M_{t-1},x_t)$, where $M_t$ is the next state. 
        \item 
        An $\op_t=\query(x_t)$  is processed by invoking $(M_t,b_t) \gets A_{\query}(r,M_{t-1},x_t)$, where the bit $b_t$ is the answer to the query.\footnote{We allow a query to modify the state, i.e., a dynamic filter may employ an unsteady representation of the dataset. }
        \end{enumerate}
\end{definition}
%%%

We define the family of dynamic filters $\Fd(u,n,\pfail,\pfp)$ that have two types of errors: initialization errors and false-positive errors. 
Note that the probability space is over the choice of the random bit string $r$. 
%%%
\begin{definition}[Semantics of a Dynamic Filter]
A dynamic filter $F$ is in $\Fd(u,n,\pfail,\pfp)$ if, for every $(u,n)$-sequence $\{\op_t\}_{t}$, the following holds:
   \begin{enumerate}
       \item 
       Success Probability. The transition to the state $\fail$ satisfies: 
       \begin{enumerate}
           \item Initialization never fails. If $\op_t=\init$, then $M_t\neq \fail$. 
           \item The state $\fail$ is a sink. If $M_{t-1}=\fail$ and $op_t\neq \init$, then $M_t=\fail$.
           \item Insertion and deletion  fail with probability at most $\pfail$.
           If $\op_t\in\brk[c]{\ins(x_t),\del(x_t)}$, then $\Pr_r[M_t=\fail \mid M_{t-1}\neq \fail] \leq \pfail$. 
           \item Queries do not fail.
        If $\op_t=\query(x_t)$, then 
       $\Pr_r[M_t=\fail \mid M_{t-1}\neq \fail] =0$.
       \end{enumerate}
       \item Completeness (no false-negatives). If $\op_t=\query(x_t)$ and  $(x_t\in S_t$ or $M_{t-1}=\fail)$, then $\Pr_r [b_t=1]=1$.
       \item Soundness (bounded false-positives). If $\op_t=\query(x_t)$,  $x_t\not\in S_t$, then $\Pr_r [b_t=1\mid M_{t-1}\neq \fail]\leq \pfp$.
   \end{enumerate}
\end{definition}
Observe that: 
\begin{enumerate*}[label=(\arabic*)]
\item To avoid false-negative errors, the response to queries is always ``yes'' (i.e., $b_t=1$) if the state is $\fail$.
    \item
The failure probability $\pfail$ in a dynamic filter is per insert/delete operation. Conversely, in an N-static filter, $\pfail$ bounds the probability of failure during initialization.
\end{enumerate*}
%%%
\begin{definition}[Dynamic Filter Space]
The \emph{space} used by a dynamic filter $F\in\Fd(u,n,\pfail,\pfp,\pfn)$ is the maximum length (in bits) of the state of $F$ over all $(u,n)$-sequences of operations and over all strings $r$. 
Formally, 
\begin{align*}
    \fspace(F)\triangleq \max_{\{op_t\}_t} \max_r \max_t |M_t|\;.
\end{align*}
\end{definition}
\medskip\noindent

The space refers only to the state that is passed as an argument to the next operation. In particular, the space does not include the string $r$ of random bits nor any internal memory used by the filter's algorithms. 

\section{Lower Bound on Space of N-Static Filter}\label{sec:N-static filter lower bound}
In this section, we prove a lower bound on the space of an N-static filter that  asymptotically equals the space of a dictionary with exact representation (provided that $u\geq 2n$, $\pfn=o(1)$, and $\pfail\leq 1/(1+ n\pfn) -2^{-n}$). 

The intuition of the proof for the space of an N-static filter is as follows. The set of all states has cardinality $2^{\fspace(F)}$. How many datasets of cardinality $n$ can an average state approximately represent? Fix a state $M$, and let $Y$ denote the subset of elements for which a query with state $M$ returns a ``yes''. The state $M$ is an approximate representation of dataset $S$ if there are no false-positives (i.e., $Y\subseteq S$)  and the number of false-negatives (i.e., $|S\setminus Y|$) is at most $  n\pfn$. In such a case, given $M$, an exact representation of $S$ can be obtained by describing the set $S\setminus Y$. 
Hence, the number of datasets $S$ is not less than the number of initial states times the number possible ``false-negative'' sets $S\setminus Y$. Informally, we want to prove that \footnote{We emphasize that this equation is incorrect and appears only to provide intuition. Indeed, the ``average'' number of false-negatives is bounded by $ n\pfn$, but, for some datatsets and some random strings, the number of false-positives might be larger. 
Moreover, the equation ignores failure during initialization.}
\begin{align*}
    2^{\fspace(F)}\cdot \binom{u}{ n\pfn}   \geq \binom{u}{n}\;.
\end{align*}

We deal with randomization and failure by defining a pair $(r,S)$ (where $r$ is the random string and $S$ is a dataset) as \emph{good} if the initial state $M(r,S)$ is not $\fail$ and the set of false-negatives $S\setminus Y$ is not too big (i.e., $|S\setminus Y|\leq \alpha\pfn \cdot n$). We prove that there exists a ``magic'' random string $r^*$ such that, for many datasets $S$, the pair $(r^*,S)$ is good.

\begin{theorem}\label{thm:static filter}
For every N-Static filter
$F\in\Fs(u,n,\pfail,\pfn)$ and for every $\alpha>1$, it holds that
\begin{align}\label{eq:static filter lb}
    2^{\fspace(F)}\cdot 
    \sum_{0 \leq k\leq \alpha n\pfn} 
    \binom{u}{k} 
    &\geq
    \left(1-\frac{1}{\alpha}-\pfail\right)\cdot \binom{u}{n} \;
\end{align}
\end{theorem}
\medskip\noindent
Throughout the following proof, $S$ denotes a subset of $n$ elements in $[u]$ and $r$ denotes a random bit string.
\begin{proof}
    Define the set false-negatives  $\FN(r,S)$ by
    \begin{align*}
        \FN(r,S)&\triangleq \{x\in S \mid A_{\query} (r,M(r,S),x) = 0\}\;.
    \end{align*}
By the soundness of an N-static filter,
\begin{align*}
    \forall x \in S, \ \Pr_r[x \in \FN(r, S)\mid M(r,S)\neq\fail] &\leq \pfn
\end{align*}
By linearity of expectation,
\begin{align*}
        \E_r [ |\FN(r,S)| ~\mid M(r,S)\neq\fail] &= \sum_{x \in S} \Pr_r[x \in \FN(r, S)\mid M(r,S)\neq \fail] \leq \pfn \cdot n \;.
    \end{align*}
Define 
\begin{align*}
        G_S &\triangleq \{r~:~ |\FN(r,S)|\leq \alpha\pfn \cdot n\}\;.
    \end{align*}
By Markov's inequality, 
\begin{align*}
\Pr_r[r\not\in G_S\mid M(r,S)\neq\fail] 
\quad\leq\quad \frac{1}{\alpha}\;. 
\end{align*}

\medskip\noindent
Define 
\begin{align*}
    C(r,S) &\triangleq \begin{cases}
        1 &\text{if $r\in G_S$ and  $M(r,S)\neq \fail$,}\\
        0 &\text{otherwise.}
    \end{cases} 
\end{align*}
It follows that, for every $S$, 
\begin{align*}
\Pr_r [C(r, S) = 1] &\geq 
1-\brk*{\Pr_r[r\not\in G_S\mid M(r,S)\neq\fail] +\Pr_r[M(r,S)=\fail]} \geq
1-\frac{1}{\alpha}-\pfail\;.
\end{align*}
By Yao's Minimax Principle, it follows that there exists an $r^*$ such that
\begin{align}\label{eq:good datasets}
| \{S : C(r^*,S) = 1\}| &\geq \brk*{1-\frac{1}{\alpha} - \pfail}\cdot \binom{u}{n} \;.
\end{align}

To complete the proof it suffices to prove that the LHS of \Cref{eq:static filter lb} is not less than the LHS of \Cref{eq:good datasets}. To this end, we define an injective function $\pi:\{ S \mid C(r^*,S)=1 \} \rightarrow \{0,1\}^*\times \mathbb{N}$ such that the cardinality of the image is not greater than the LHS of \Cref{eq:static filter lb}.

The function $\pi(S)$ is defined as follows. 
The first coordinate of $\pi(S)$ equals $M(r^*,S)$.
It is useful to view $M(r^*,S)$ as an exact encoding of $Y\triangleq S\setminus \FN(r^*,S)$. That is, $Y$ is the set of all elements $x\in S$ for which $A_{\query}(r^*,M(r^*,S),x)=1$. By Definition, there are at most $2^{\fspace(M)}$ possible values for $M(r^*,S)$.

Consider a lexicographic ordering of all the subsets of $[u]\setminus Y$ of cardinality at most $\alpha n\pfn$. Because $C(r^*,S)=1$, it follows that $\FN(r^*,S)=S\setminus Y$ is one of these subsets. 
The second coordinate of $\pi(S)$ is defined to be the index $i$ of $\FN(r^*,S)$ among these subsets. 
This index $i$ is at most $\sum_{0 \leq k\leq \alpha n\pfn} \binom{u}{k}$. This completes the definition of the function $\pi$ and the bound on its image.

The function $\pi$ is injective because one can reconstruct $S$ from $\pi(S)$. 
Indeed, $M(r^*,S)$ encodes $Y$ and the index $i$ encodes $S\setminus Y$, and the theorem follows.
\end{proof}

Consider a fraction $0\leq \beta \leq 1$ and a universe of cardinality $u$. 
Consider two experiments:
\begin{enumerate*}[label=(\arabic*)]
    \item pick an arbitrary dataset of $n$ elements from the universe, and
    \item pick a sequence of $1/\beta$ arbitrary datasets of $\beta\cdot n$ elements from the universe.
\end{enumerate*}
The following claim bounds the number of outcomes of the second experiment (i.e. $\binom{u}{\beta n}^{1/\beta}$) by the number of outcomes of the first experiment (i.e, $\binom{u}{n}$).
\begin{claim}\label{claim:binom}
    For every $0\leq \beta \leq 1$, it holds that 
    \begin{align}\label{eq:ratio binoms}
    \log \binom{u}{\beta n}
&\leq
\beta\log \binom{u}{n}+O(n)\;.
    \end{align}
\end{claim}
\begin{proof}
    \begin{align}
        \frac{\binom{u}{n}}{\binom{u}{\beta n}} &\geq
        \frac{\brk*{\frac{u}{n}}^n}{\brk*{\frac{eu}{\beta n}}^{\beta n}}\nonumber \\
        &=\brk*{\frac{eu}{n}}^{(1-\beta)n}\cdot \brk*{\frac{\beta^\beta}{e}}^{n} \nonumber \\
        &\geq 
        \binom{u}{n}^{1-\beta} \cdot \brk*{\frac{\beta^\beta}{e}}^{n}\label{eq:for coro}
        \;.
    \end{align}
By rearranging, we get
\begin{align*}
    \binom{u}{n}^{\beta} \cdot \brk*{\frac{e}{\beta^\beta
    }}^{n} \geq \binom{u}{\beta n}
\end{align*}
By taking the logarithm, it follows that 
 \begin{align*}
     \log \binom{u}{\beta n}&\leq \beta\cdot\log \binom{u}{n} +n\cdot \log \brk*{\frac{e}{\beta^\beta}}\;.
 \end{align*}
 and the claim follows.
\end{proof}

%%%%%%%%%%

\begin{corollary} \label{binom-claim-rephrase}
    If  $\pfn \leq 1-\frac{1}{n}$, then
    \begin{align*}
    \log\frac{\binom{u}{n}}{\binom{u}{1+ n\pfn}} &\geq  
    (1 - \pfn - \frac1n) \cdot \log \binom{u}{n} - O(n)\;.
    \end{align*}
\end{corollary}
\begin{proof}
Apply \Cref{claim:binom} (i.e., \Cref{eq:for coro}) with $\beta\triangleq\pfn+\frac{1}{n}\leq 1$.
\end{proof}

The following corollary states the conditions under which the space of an N-static filter that lacks false-positive errors is at least $1-o(1)$ times (or big-Omega) the space of a dictionary with exact representation (i.e., $\log \binom{u}{n}$).

\begin{corollary}[Lower bound on space of N-filter]\label{coro:pfp=0}
If $\pfn < 1-\frac{1}{n}$, $u\geq 2n$ and 
$\pfail\leq \frac{1}{1+n\pfn}-2^{-n}$,
then for every N-Static filter $F\in\Fs(u,n,\pfn,\pfail)$,
 \begin{align}
     \label{eq:log u choose n}
     \fspace(F)&\geq (1-\pfn-\frac 1n) \cdot \log\binom{u}{n} - O(n)
\;.
 \end{align}
Moreover, if $\pfn=o(1)$, then
\begin{align}
     \label{eq:space N-filter}
     \fspace(F)&\geq (1-o(1)) \cdot \log\binom{u}{n}
\;.
 \end{align}
\end{corollary}
\begin{proof}
Define $\alpha\triangleq 1+ \frac{1}{n\pfn}$. 
Note that $\alpha\pfn=\pfn+\frac{1}{n} <1$ and $\alpha n\pfn < u/2$. We increase the LHS of \Cref{eq:static filter lb}, that is,  
    \begin{align}\label{eq:static proof LHS}
    2^{\fspace(F)}\cdot 
    (\alpha n\pfn  + 1) \binom{u}{\alpha n\pfn}
    &\geq
    2^{\fspace(F)}\cdot 
    \sum_{k\leq \alpha n\pfn} 
    \binom{u}{k} 
    \end{align}

We combine  \Cref{eq:static filter lb} and \Cref{eq:static proof LHS}, rearrange and apply a logarithm to both sides to obtain
\begin{align*}
    \fspace(F) & \geq \log \frac{\binom{u}{n}}{\binom{u}{\alpha n\pfn}} -\log(\alpha n\pfn + 1) + \log (1-\frac{1}{\alpha}-\pfail)\\
    &\geq
    (1-\pfn-\frac 1n)\cdot \log \binom{u}{n} - O(n) &\text{(by \Cref{binom-claim-rephrase} and $(1-\frac{1}{\alpha}-\pfail\geq 2^{-n})$)}
\end{align*}
and the corollary follows.
\end{proof}

%%%%%%%%%%%%%
\section{Lower Bound on the Space of Dynamic Filters with Duplicate Insertions or Deletions of Nonelements}
In this section, we prove a lower bound on the space of dynamic filters that permit duplicate-insertions or deletions-of-nonelements.
The lower bound on the space of such a filter is big-$\Omega$ of the space required for exact representation  (the constraints on the parameters appear in \Cref{thm:dupilicate ins deletion of nonelement}).

\begin{theorem}\label{thm:dupilicate ins deletion of nonelement}
If a dynamic filter $F=\Fd(u,n,\pfail,\pfp)$ is correct with respect to $(u,n)$-sequences and the parameters satisfy
\begin{enumerate*}[label=(\roman*)]
    \item $2n\cdot\pfail\leq \frac{1}{1+n\pfp}-2^{-n}$,
    \item $\eps^+ \leq 1-\frac{1}{n}$, and
    \item $u \geq 2n$, 

\end{enumerate*}
 then 
 \begin{align*}
     \fspace(F)&\geq \frac{1}{2}\cdot \brk*{1-\pfp-\frac 1n} \cdot \log\binom{u}{n} - O(n)
\;.
 \end{align*}
Moreover, if $\pfp\leq 1-\Omega(1)$, then
 \begin{align*}
     \fspace(F)&\geq \Omega\brk*{\log\binom{u}{n}}
\;.
 \end{align*}
\end{theorem}

We remark that the proof of \Cref{thm:dupilicate ins deletion of nonelement} does not assume that the dynamic filter supports duplicate-insertions (but it does assume that deletions-of-nonelements are supported). 
The following claim states that deletions-of-nonelements can be reduced to duplicate-insertions, and vice-versa, by increasing the number of operations by a factor of at most two. 
We refer to an $\ins(x)$ as a \emph{potential} dupilicate-insertion if the element $x$ may be an element in the current dataset (but might also be a nonelement). 
Similarly, a $\del(x)$ is a \emph{potential} deletion-of-a-nonelement, if $x$ may be a nonelement.

\begin{claim}[Reduction between duplicate-insertions and deletions-of-nonelements]\label{claim:reductions}
One can reduce potential duplicate-insertions to potential deletions-of-nonelements and vice versa. 
\end{claim}
\begin{proof}
A potential duplicate-insertion $\ins(x)$ operation can be reduced to a potential deletion-of-a-nonelement by performing a $\del(x)$ operation just before the $\ins(x)$ operation.

A potential deletion-of-a-non-element $\del(x)$ operation can be reduced to a potential duplicate-insertion by performing an $\ins(x)$ operation just before the $\del(x)$ operation. \footnote{The insert-before-delete may cause the filter to overflow (i.e., on a deletion-of-a-nonelement with a dataset of cardinality $n$, the cardinality of the dataset after the insertion is $n+1$). To avoid an overflow due to the extra insertion, the  filter should be initialized to support datasets of cardinality $n+1$. }
\end{proof}
%%%%
\subsection{Preliminaries}
%%%%

\paragraph{Notation.} Consider a dynamic filter $F$.
We focus on sequences of operations consisting of inserting a dataset $S$, deleting a subset $T\subseteq S$, and possibly querying an element. 
To simplify notation, we use the following abbreviations ($r$ denotes the random string of bits chosen by the filter):
\begin{enumerate}
    \item $M(F,r,+S)$ is the state of $F$ after initializing $F$ and inserting the elements of $S$ in ascending order. 
    \item $M(F,r,+S,-T)$ is the state of $F$ after initializing $F$, inserting the elements of $S$ in ascending order and deleting the elements of $T$ in ascending order. 
    \item For a state $M$, define the ``yes''-set by
     \begin{align*}
         Y(F,r,M) &\triangleq \{x\in [u] \mid A_{\query}(r,M,x) \text{ answers } 1 \}\;.
     \end{align*}
\end{enumerate}
We  omit the filter $F$ or the string $r$ in this notation if the discussion deals with a specific filter or a specific value of $r$. For example, $M(+S)$ is an abbreviation of $M(F,r,+S)$.

\begin{definition}[Witness-Based Filter]\label{def:witness-based}
A filter $F$ is \emph{witness-based} if it satisfies the following property. For every $x\in [u]$,  $x\in Y(F,r,M)$ if and only if there exists a dataset $S\subseteq [u]$ such that: 
\begin{enumerate*}[label=(\arabic*)]
    \item $|S|=n$, 
    \item $x\in S$, and
    \item $M(F,r,+S)=M$.
\end{enumerate*}
We refer to such an $S$ as a \emph{witness} for $x$.
\end{definition}

\begin{claim}\label{claim:witness-based}
    Every dynamic filter $F$ can be transformed to a witness-based filter $F^*$ with the same parameters and space.
\end{claim}
\begin{proof}
    The difference between $F$ and $F^*$ is only in the algorithm for queries. On a query, the filter $F^*$ exhaustively searches for a witness $S$ (the extra space needed for the exhaustive search is temporary and is not counted as part of the space of the filter). The answer is  ``yes'' (i.e., $b=1$) if and only if a witness is found.
    The next state $M_{t+1}$ is the next state that $F$ computes for the query. 
\end{proof}
Note that the ``yes''-set of $F^*$ is contained in the ``yes''-set of $F$. That is, for every $r$ and every state $M$, $Y(F^*,r,M)\subseteq Y(F,r,M)$.

%%%%%%%%%%%
\subsection{Proof of 
%\protect
\texorpdfstring
{\Cref{thm:dupilicate ins deletion of nonelement}}{Theorem 14}}

The following lemma states that if a witness-based dynamic filter permits deletions-of-nonelements, then every false-positive element after inserting the elements of $S$ is also a false-positive element after inserting $S$ and deleting $S$. 
\begin{lemma}[Sticky False-Positives Lemma]\label{lemma:false positives}
For every witness-based dynamic filter that is correct with respect to $(u,n)$-sequences, it holds that, for every dataset $S$ and every string $r$,
\begin{align}\label{eq:false positive}
    Y(M(+S)) ~\setminus~ S
    &\quad\subseteq\quad
    Y(M(+S,-S))\;.
\end{align}
\end{lemma}
\begin{proof}
    Consider an element $x\in Y(M(+S))\setminus S$. Because the filter is witness-based and because the answer to query $A_{\query}(r,M(+S),x)$ is ``yes'', it follows that there is a witness $T$ for $x$. 
    By definition, $M(+T)=M(+S)$, and hence $M(+T,-S)=M(+S,-S)$.
    Because $x\in T\setminus S$, it follows that $x\in Y(M,+T,-S)$, and the lemma follows.
\end{proof}

\begin{definition}
    [Reduction from N-static-filter to P-dynamic-filter]\label{def:reduction}
    Given a witness-based dynamic filter $F$, define the N-Static filter
    $\sigma(F)$ as follows:
  \begin{enumerate}
      \item Initialization. For a string $r$ and a dataset $S\subseteq [u]$ of cardinality $n$, the initial state of $\sigma(F)$ is the pair $M(\sigma(F),r,S)\triangleq (M(F,r,+S),M(F,r,+S,-S))$.
\item 
On $\query(x)$, filter $\sigma(F)$ computes two bits $b_1,b_2$, defined by:
\begin{align*}
    b_1 =1 & \Longleftrightarrow x\in Y(F,r,M(F,r,+S))\\
    b_2 =1 & \Longleftrightarrow x\in Y(F,r,M(F,r,+S,-S)) \;.
\end{align*}
The answer of $\sigma(F)$ to $\query(x)$ is $b\triangleq b_1 \wedge \mathsf{not} (b_2)$.
  \end{enumerate}
\end{definition}

\begin{lemma}[Reduction correctness]\label{lemma:reduction}
    If $F\in\Fd(u,n,\pfail,\pfp)$ is witness-based, then 
$\sigma(F)\in\Fs(u,n,\pfail',\pfn)$, where
    \begin{enumerate}
        \item $\pfail'= 2n\cdot \pfail$,
        \item $\pfn = \pfp$, and
        \item $\fspace(\sigma(F))= 2\cdot \fspace(F)$.
    \end{enumerate}
\end{lemma}
\begin{proof}
The filter $\sigma(F)$ satisfies the following properties:
\begin{enumerate}
    \item No false-positives.
    Suppose that  $x\not\in S$. Assume towards a contradiction that $b=1$. Hence $b_1=1$ and $x\in Y(F,r,M(F,r,+S))$. By \Cref{lemma:false positives}, $x\in Y(F,r,M(F,r,+S,-S))$, hence $b_2=1$, a contradiction.
    \item False-negative probability bounded by $\pfp$.
    Suppose that $x\in S$. Then $x\in Y(F,r,M(F,r,+S))$ and $b_1=1$. It follows that $\Pr_r[b=0]=\Pr_r[b_2=1]$. However, $x\not\in S\setminus S = \emptyset$. Hence, $\Pr_r[b_2=1]\leq \pfp$. 
    \item Failure probability bounded by $2n\cdot\pfail$. A failure may occur during initialization when the states $M(F,r,+S)$ and $M(F,r,+S,-S)$ are computed. These computations involve $n$ insertions and $n$ deletions, each of which may fail with probability $\pfail$. By a union bound, $\pfail'\leq 2n\cdot \pfail$.
    \item Space bounded by $2\cdot\fspace(F)$. The states of filter $\sigma(F)$ consists of pairs of states of filter $F$, hence $\fspace(\sigma(F))\leq 2\cdot \fspace(F)$.
    \end{enumerate}
\end{proof}
The proof of \Cref{thm:dupilicate ins deletion of nonelement} follows from \Cref{claim:witness-based}, \Cref{lemma:reduction} and \Cref{coro:pfp=0}.
%%%%%%%%%%%%%
%\bibliographystyle{alpha}
\newcommand{\etalchar}[1]{$^{#1}$}

%\bibliography{filter}
%\appendix

\end{document}